\documentclass[11pt]{article}
\usepackage{amssymb}
\usepackage{fullpage}
\usepackage{amsfonts}
\usepackage{amsmath}
\usepackage{cancel}
\usepackage{color}
\usepackage{fancyhdr}
\usepackage{ mathrsfs }
\usepackage{latexsym}
\usepackage{amsmath,amssymb,amsthm}
\usepackage{epsfig}
\usepackage{tikz}
\usetikzlibrary{arrows,shapes,positioning,automata}

\newtheorem{theorem}{Theorem}

\begin{document}

\title{A Note on the Assignment Problem with Uniform Preferences}

\author{
Jay Sethuraman
\thanks{IEOR Department, Columbia University, New York,
NY;
{\tt jay@ieor.columbia.edu}. Research supported by NSF grant
CMMI-0916453 and CMMI-1201045.}
\and
Chun Ye
\thanks{IEOR Department, Columbia University, New York,
NY;
{\tt cy2214@columbia.edu}}
}

\date{\today}

\maketitle

\begin{abstract}
Motivated by a problem of scheduling unit-length jobs with weak preferences over time-slots, the random assignment problem (also called the house allocation problem) is considered on a uniform preference domain. For the subdomain in which preferences are strict except possibly for the class of unacceptable objects, Bogomolnaia and Moulin characterized the probabilistic serial mechanism as the only mechanism satisfying equal treatment of equals, strategyproofness, and ordinal efficiency. The main result in this paper is that the natural extension of the probabilistic serial mechanism to the domain of weak, but uniform, preferences fails strategyproofness, but so does every other mechanism that is ordinally efficient and treats equals equally. If envy-free assignments are required, then any (probabilistic or deterministic) mechanism that guarantees an ex post efficient outcome must fail even a weak form of strategyproofness.
\end{abstract}

\section{Introduction}
We study the assignment problem, which is concerned with allocating a set of objects to a set of agents, each of whom wishes to receive at most one object.  Agents have preferences over the objects, and the goal is to allocate the objects to the agents in a fair and efficient manner. Further, as each agent's preference ordering over the objects is private information, we require the mechanism to be strategyproof: it should be a dominant strategy for the agents to report their preference ordering truthfully. If the objects are divisible, we can think of a fractional assignment in which an object may be allocated in varying amounts to multiple agents so that the total amount allocated of any object  is at most 1, and so that each agent receives at most one unit in all. If the objects are indivisible, one can think of a lottery over assignments, which again results in a fractional assignment matrix in which entry $(i,a)$ represents the probability that agent $i$ receives object $a$. These two views are equivalent for our purposes; while in the rest of the paper we assume that the objects are indivisible, all of our results extend to the case of divisible objects.
There is now a rich literature on such models with applications to many real-life allocation problems including allocating students to schools in various cities, the design of kidney exchanges, etc \cite{AbdulkadirogluSonmez1998,AbdulkadirogluSonmez2003,Budishetal2013,Rothetal2005}. The two prominent mechanisms that have emerged from this literature are the {\em Random Priority} (RP) mechanism and the {\em Probabilistic Serial} (PS) mechanism. The PS mechanism is stronger in terms of its efficiency and equity properties, but it is only weakly strategyproof in the strict preference domain and not strategyproof in the full preference domain; whereas the RP mechanism is strategyproof, but satisfies only a weaker version of efficiency and envy-freeness. Furthermore, Bogomolnaia and Moulin~\cite{BogomolnaiaMoulin2001} show that no strategyproof mechanism can satisfy the stronger form of efficiency and equity that the PS mechanism satisfies.

This paper is inspired by the paper of Bogomolnaia and Moulin~\cite{BogomolnaiaMoulin2002} that characterized the PS mechanism on a restricted preference domain. The PS mechanism was introduced in an earlier paper of Cres and Moulin~\cite{cresmoulin1999} that was motivated by the problem of scheduling unit-length jobs with deadlines. Suppose there are $n$ jobs, each requiring a unit processing time, and all jobs are available at time zero. As the jobs all have unit-length, one could think of the scheduling problem as one of assigning time-slots $1, 2, \ldots, n$ to the jobs, so that slot $k$ represents the interval $(k-1,k]$, and a job assigned to slot $k$ finishes at time $k$.
Jobs have deadlines and earn a non-negative utility if they complete before their deadline. Specifically, if the deadline of job $j$ is $d_j$, then the utility of assigning $j$ to slot $k$ is monotonically decreasing in $k$ until the deadline, after which it drops to zero. That is, if
$u_{j,k}$ denotes the utility of assigning job $j$ to slot $k$, then
$$u_{j,1} > u_{j,2} \ldots > u_{j,d_j} > 0 = u_{j,d_{j}+1} = u_{j,d_{j}+2}, \ldots, u_{j,n}.$$
The goal is to use a nonpricing mechanism to schedule the jobs in a fair and efficient manner based on their reported utility information. Cres and Moulin~\cite{cresmoulin1999} proposed the PS mechanism and showed that it finds an ordinally efficient and envy-free allocation (all definitions appear in the next section); furthermore, they showed that the PS mechanism is strategyproof  on this domain: note that each job/agent need only report their deadline, and they show that it is a weakly dominant strategy for each job to report its deadline truthfully. Bogomolnaia and Moulin~\cite{BogomolnaiaMoulin2002} characterize the PS mechanism on this restricted domain in two different ways: first, they show that ordinal efficiency and envy-freeness characterize the PS outcome on this restricted domain; and second, they show that it is the only strategyproof mechanism that is ordinally efficient and treats equals equally. Taken together, their result shows that the PS mechanism is perhaps the only compelling mechanism on this restricted preference domain\footnote{Cres and Moulin~\cite{cresmoulin1999} show that the PS mechanism is in fact group strategyproof, although this stronger property is not needed in the characterization results of PS.}.

In this paper we consider a slightly more general domain, again inspired by the problem of scheduling unit-length jobs. For simplicity, assume there are $n$ agents and $n$ objects, and suppose the objects are arranged in the order $(1,2,\ldots, n)$ by all the agents. Each agent's preference ranking, however, is determined by a {\em weakly} decreasing utility function over the objects, in contrast to a strictly decreasing utility function over the objects till a deadline. (A good way to visualize this preference domain is to have each agent separate the sequence of objects into indifference classes, without disturbing the common order on the objects.)
This domain is quite natural in the scheduling context, where completing a job early is always (weakly) better, but  jobs may be insensitive to completion times within a certain time interval, and these intervals may change from job to job.
The domain considered in the earlier papers is a special case in which, for each agent, all but the final indifference class has a single object.
It is then natural to ask if the two characterizations of PS extend to this domain. It turns out that the answer is negative for both characterizations. We show that the PS outcome (actually, a correspondence) is no longer the only outcome that is ordinally efficient and envy-free, nor is the PS mechanism strategyproof on this domain. Somewhat surprisingly, we show that:
\begin{itemize}
\item No weakly strategyproof mechanism can satisfy both ex post efficiency and envy freeness on this domain, when there are three or more agents; and
\item No strategyproof mechanism can satisfy both ordinal efficiency and equal treatment of equals on this domain, when there are four or more agents. 
\end{itemize}
 
The literature on random assignment problems focuses on simultaneously satisfying various notions of fairness, efficiency, and strategyproofness, and several impossibility results have been established over the last two decades~\cite{Aziz14, BogomolnaiaMoulin2001, Chambers2004, HyllandZeckhauser1979, KattaSethuraman2006, Zhou1990}.
Our two main impossibility results are strengthened versions of similar results in the literature in which preferences are drawn from richer domains. Specifically, versions of the two impossibility results have been obtained by Katta and Sethuraman~\cite{KattaSethuraman2006} on the full preference domain (where any weak ordering of the objects is permissible), and by Bogomolnaia and Moulin~\cite{BogomolnaiaMoulin2001} on the strict preference domain (where any {\em strict} ordering of the objects is permissible). Thus the surprising element in our result is that these difficulties
persist even in domains in which the preferences are severely restricted.

Our work contributes to the rich and growing literature on matching and allocation problems in which monetary transfers are not permitted. The PS mechanism and the {\rm Random Priority} mechanims are central mechanisms for such allocation
problems and have been studied extensively from several points of view, see the recent survey of Sonmez and Unver~\cite{sonmezunver2011} for an overview. This has also inspired other characterizations and extensions of the
PS mechanism~\cite{AthanassoglouSethuraman2011, Aziz14, BogomolnaiaHeo2012, Hashimotoetal2014,  Kojima2009, KojimaManea2010}. There is an equally extensive literature on models where monetary transfers are allowed to restore fairness or strategyproofness in a queueing or scheduling setting~\cite{Dolan1978,Mendelson1985,MendelsonWhang1990,Naor1969,Sujis1996}, and we refer the reader to the work of Hassin and Haviv~\cite{hh2003} for a comprehensive overview.

\section{Preliminaries}

\subsection{Model and Definitions}
An assignment problem is given by a triple $(N,O, \succsim)$, where $N = \{1, \ldots, n\}$ is the set of agents, $O = \{o_1, \ldots o_n\}$ is the set of objects, and the preference profile $\succsim = (\succsim_1, \ldots, \succsim_n)$ specifies each agent's preference ordering over the objects. We will assume that the preference relation of each agent is complete (every pair of objects is comparable) and transitive. By $a \succsim_i b$, we mean that  agent $i$ weakly prefers object $a$ to object $b$. We write $a \succ_i b$ if $i$ strictly prefers $a$ to $b$, i.e. $a \succsim_i b$ but $b \not \succsim_i a$; and we use $a \sim_i b$ when $i$ is indifferent between $a$ and $b$, i.e. $a \succsim_i b$ and $b \succsim_i a$. We assume that the indifference relation is also transitive. Thus each agent has a most-preferred subset of objects (and the agent is indifferent between all the objects within this set), followed by a most-preferred subset of objects among the remaining ones, etc.

In this paper, we shall consider the {\em uniform} preference domain in which
$o_1 \succsim_i o_2 \succsim_i \ldots \succsim_i o_n$ for every agent $i \in N$. Agents differ in their preference ordering only in their strict preference relation $\succ_i$ (and hence their indifference relation $\sim_i$). In the rest of the paper, we use the following notation for the preference ordering of the agents: all the objects within an indifference class for an agent appear within braces in that agent's preference list, and these maximal indifference classes are separated by a comma; objects are always written in subscript order; and the braces are omitted for singleton indifference classes. Thus, the preference ordering 
\[o_1 \succ_i o_2 \sim_i o_3 \sim_i o_4 \succ_i o_5\]
for agent $i$ is written as
\[i: o_1, \{o_2 \ o_3 \ o_4\}, o_5.\] 

By a mechanism, we mean a mapping from the set of all preference profiles (within this restricted domain) to a doubly stochastic matrix\footnote{If the number of agents is not the same as the number of objects, the outcome is typically a matrix in which the row-sums and the column-sums are each {\em at most} 1. One can always balance such a problem by adding dummy agents or dummy objects.}, which we call the assignment matrix for that profile. The assignment matrix is {\em deterministic} if its entries are $\{0,1\}$ (and so the outcome is a {\em matching} of the agents and objects); otherwise, it is {\em probabilistic}. If a mechanism maps each preference profile to a deterministic matrix, the mechanism is deterministic; otherwise the mechanism is probabilistic\footnote{Alternatively, we could have defined a probabilistic mechanism as a lottery over deterministic mechanisms. In this view, different lotteries are regarded as different
mechanisms, even if they result in the same assignment matrix for each preference profile.}.
As a consequence of the Birkhoff-von Neumann theorem \cite{birkhoff1946}, the outcome of a probabilistic mechanism can be implemented as a lottery over deterministic assignments.

Given two probabilistic assignments $P$ and $Q$, we say that agent $i$ prefers $P$ to $Q$ if $P_i$\footnote{$P_i$ denotes the $i$-th row of $P$} stochastically dominates $Q_i$ according to $i$'s preferences. Formally,
\[P_i \succsim_i Q_i  \Longleftrightarrow \sum_{k:k \succsim_i j} p_{ik} \geq  \sum_{k:k \succsim_i j} q_{ik}, \quad \forall j \in O.  \]
We say that $i$ strictly prefers $P$ to $Q$, denoted by $P_i \succ_i Q_i$, if at least one of the inequalities in the above definition is strict. Note that this definition is only a {\em partial} order, as an agent may not be able to compare two probabilistic allocations.
Finally, we say that $P$ \emph{stochastically dominates} $Q$, denoted by $P \succsim Q$, if $P_i \succsim_i Q_i$ for all $i \in N$, with $P_i \succ_i Q_i$ for some $i \in N$. Again, this notion of stochastic dominance defines a partial order on the set of doubly stochastic matrices.

\subsection{Desirable Properties}
We define some desirable properties of mechanisms that play an important role in the rest of the paper.

\paragraph{Ordinal Efficiency.}
An assignment matrix $P$ is \emph{ordinally efficient} if it is not stochastically dominated by any other random assignment matrix $Q$ such that $Q \succsim P$. It is well known that any ordinally
efficient matrix can be implemented as a lottery over deterministic Pareto efficient assignments. Furthermore, checking whether or not a given assignment matrix is ordinally efficient is computationally easy~\cite{BogomolnaiaMoulin2001, KattaSethuraman2006}.

\paragraph{Ex post Efficiency}
A weaker notion of efficiency that we will consider is ex post efficiency, which is satisfied by the random priority mechanism. A bi-stochastic matrix $P$ is \emph{ex post efficient} if it can be written as a convex combination of Pareto efficient assignments. 

\paragraph{Envy-Freeness.}
An assignment matrix $P$ is \emph{envy free} if the probabilistic assignment of every agent $i$ stochastically dominates the probabilistic assignment of every other agent with respect to agent $i$'s preference ordering. Let $P_i$ denote the probabilistic assignment of agent $i$ in the matrix $P$. Then, $P$ is envy-free if $P_i \succsim_i P_{i'}$ for all $i, i' \in N$. 

\paragraph{Equal Treatment of Equals.}
An assignment matrix $P$ satisfies \emph{equal treatment of equals} if agents with identical preferences get equivalent allocations. Formally, $P$ satisfies equal treatment of equals if for all $i, i' \in N$ such that $\succsim_i = \succsim_{i'} = \succsim$, we have
$$\sum_{k:k \succsim j} p_{ik} = \sum_{k:k \succsim j} p_{i'k}, \;\; \forall j \in O.$$

 
\paragraph{Strategyproofness.}	
The properties defined so far pertain to the outcome on a single profile; strategyproofness, however, is a property of the mechanism, in particular, on how the mechanism behaves on pairs of profiles in which all but one of the agents report the same preference ordering. A mechanism is \emph{strategyproof} if it is a weakly dominant strategy for each agent to report their true preference ordering. Formally, a mechanism is strategyproof if 
$$P_i(\succsim_i, \succsim_{-i}) \; \succsim_i \; P_i(\succsim'_i, \succsim_{-i}),$$
for all agents $i \in N$, and for all preference profiles $\succsim_{-i}$ of the other agents, and for every pair of preferences $\succsim_i, \succsim'_i$ that $i$ could report.
A random assignment mechanism is \emph{weakly strategyproof} if for each $i \in N$, and for each preference profile $\succsim_{-i}$ of the other agents, there does not exist preference ordering $\succsim'_i$  such that
$P_i(\succsim'_i, \succsim_{-i}) \succ_i P_i(\succsim_i, \succsim_{-i})$. In a strategyproof mechanism, the assignment under truthful reporting stochastically dominates the assignment under any other report; in a weakly strategyproof mechanism, however,  reporting her preference ordering truthfully will not result in an assignment that is stochastically dominated by the assignment under any other report.
It is clear from the definitions that strategyproofness implies weak strategyproofness, but not vice-versa.

\subsection{The Extended Probabilistic Serial Mechanism} 

We end this section with a very brief description of the EPS mechanism~\cite{KattaSethuraman2006}. The EPS mechanism, like the PS mechanism, can be described as a ``cake-eating'' mechanism in which agents consume their best object(s) at unit rate. Roughly, each agent  \emph{simultaneously} consumes her ``best set'' of available objects at a unit rate at each point in time. If all the preferences are strict, this determines a unique allocation for the agents; when agents have indifferences, this mechanism is not well-defined as each agent has a choice on how her unit rate is apportioned across the objects in her best set of objects. For instance, if agent $i$ strictly prefers $a$ to $b$, whereas agent $i'$ is indifferent between $a$ and $b$, letting both agents consume $a$ initially will result in each agent getting $1/2$ of $a$ and $1/2$ of $b$, which is clearly inefficient in the ordinal sense; if $i'$ consumes $b$ at rate 1, however, the outcome is ordinally efficient.
Building on this intuition, Katta and Sethuraman \cite{KattaSethuraman2006} proposed the EPS mechanism that:
\begin{enumerate}
\item Identifies a subset $S^{\star}$  of agents with the least collective claim over the union of their best objects $C(S^{\star})$ (in terms of average claim per agent within the subset); (We will refer to $S^{\star}$ as the bottleneck set.)
\item Assigns each agent in $S^{\star}$ an amount of $\frac{|C(S^{\star})|}{|S^\star|}$ of their favorite object(s);
\item Promises the rest of the agents an amount of at least $\frac{|C(S^{\star})|}{|S^\star|}$ of their favorite object(s); and
\item Removes the allocated objects, and recurses on the subproblem (agents in $S^{\star}$ now start consuming their favorite objects(s) out of the remaining objects.) 
\end{enumerate}
The authors showed that the bottleneck sets can be identified by solving a sequence of parametric max flow problems. We refer the reader to their paper for a complete description of the algorithm.

Note that in the full preference domain, an agent is insensitive to different probabilistic allocations of objects within the same indifference class as long as the allocations sum up to the same quantity for every indifference class. This motivates the following equivalence relation over the set of assignment matrices. Given a preference profile $\succsim$, let $\mathcal{I}_i$ be the collection of indifference classes of objects for agent $i$. For every $I \in \mathcal{I}_i$, let $p_{iI} = \sum_{o_j \in I} p_{ij}$. We say that two random assignment matrices $P$ and $Q$ are equivalent if and only if 
\[p_{iI} = q_{iI} \quad \; \forall i \in N, I \in \mathcal{I}_i.\]  
One can check that this defines an equivalence relation on the set of assignment matrices. An assignment matrix is an {\em  EPS assignment} if it is equivalent to the random assignment found by the EPS mechanism\footnote{Katta and Sethuraman~\cite{KattaSethuraman2006} do not discuss how Steps 2 and 3 described earlier are implemented; each allocation satisfying the conditions of Steps 2 and 3 may give a different assignment matrix but these are all equivalent.}.

\section{Main Results}
Bogomolnaia and Moulin \cite{BogomolnaiaMoulin2002} showed that if the preference domain is further restricted so that the acceptable set of objects for each agent $i$ is the set $\{o_1, o_2 \ldots, o_{k_i}\},$ and if the agents have strict (and uniform) preferences over their acceptable objects, then the PS outcome is characterized by ordinal efficiency and envy-freeness, and that it is the only strategyproof mechanism that guarantees ordinal efficiency and equal treatment of equals. We show that neither one of these results holds when the agents have weak preferences.

\subsection{Non-uniqueness of Ordinally Efficient and Envy Free Assignments}
The EPS mechanism finds an equivalence class of ordinally efficient and envy free assignments for each preference profile. However, there are other assignments with these properties.
For the preference profile

\[\begin{tabular}{lllll}
1:& $o_1$, & $\{o_2$& $o_3\}$,& $o_4$ \\
2:& $o_1$, & $\{o_2$& $o_3\}$,& $o_4$ \\
3:& $\{o_1$ & $o_2\}$,& $o_3$,& $o_4$ \\  
4:& $\{o_1$ & $o_2\}$,& $o_3$,& $o_4$ \\ 
\end{tabular}
\]
the following assignment 
\[
\begin{tabular}{lllll}
& $o_1$& $o_2$ & $o_3$ & $o_4$ \\
1:& $\frac{1}{4}$ & $0$& $\frac{1}{2}$& $\frac{1}{4}$ \\
2:& $\frac{1}{4}$ & $0$& $\frac{1}{2}$& $\frac{1}{4}$ \\
3:& $\frac{1}{4}$ & $\frac{1}{2}$& $0$& $\frac{1}{4}$ \\  
4:& $\frac{1}{4}$ & $\frac{1}{2}$& $0$& $\frac{1}{4}$ \\ 
\end{tabular}
\]
is ordinally efficient and envy free.
However, the EPS mechanism will not compute the above assignment since agents $1$ and $2$ strictly prefers $o_1$ to $o_2$ whereas agents $3$ and $4$ are indifferent between $o_1$ and $o_2$. Thus, in the EPS mechanism, agents $3$ and $4$ consume $o_2$ first so as to not compete with agents $1$ and $2$ for their unique best object. Consequently, EPS finds the following assignment
\[
\begin{tabular}{lllll}
& $o_1$& $o_2$ & $o_3$ & $o_4$ \\
1:& $\frac{1}{2}$ & $0$& $\frac{1}{4}$& $\frac{1}{4}$ \\
2:& $\frac{1}{2}$ & $0$& $\frac{1}{4}$& $\frac{1}{4}$ \\
3:& $0$ & $\frac{1}{2}$& $\frac{1}{4}$& $\frac{1}{4}$ \\  
4:& $0$ & $\frac{1}{2}$& $\frac{1}{4}$& $\frac{1}{4}$ \\ 
\end{tabular}
\] 
Clearly the two assignments do not belong to the same equivalence class: agents 1 and 2 strictly prefer the latter, whereas agents 3 and 4 strictly prefer the former.

\subsection{Impossibility Results}

\begin{theorem} \label{Imp: OE, EF wSP}
For $n \geq 3$, any mechanism that is both ex-post efficient and envy-free is not even weakly strategyproof in the uniform preference domain.
\end{theorem}
\begin{proof}

We first show the impossibility result for $n=3$. Consider Profile 1 (below). Clearly, the set of envy-free (EF) assignments at this profile is as described for some $0 \leq y \leq 1/6$.
\[\begin{tabular}{llll}
\multicolumn{4}{l}{Profile 1} \\
1:& $o_1$, & $o_2$,& $o_3$\\
2:& $o_1$, & $\{o_2$ & $o_3\}$\\  
3:& $o_1$, & $o_2$,& $o_3$\\ 
\end{tabular}
\quad 
\begin{tabular}{llll}
& $o_1$& $o_2$ & $o_3$ \\
1:& $\frac{1}{3}$ & $\frac{1}{2}-y$& $\frac{1}{6}+y$\\
2:& $\frac{1}{3}$ & $2y$& $\frac{2}{3}-2y$\\
3:& $\frac{1}{3}$ & $\frac{1}{2}-y$& $\frac{1}{6}+y$\\  
\end{tabular}
\]

By the structure of the preferences in Profile 1, agent 2 cannot receive object $o_2$ in any Pareto efficient assignment, as there is always a Pareto improvement with the agent who is assigned $o_3$. Thus $y = 0$ in any ex-post efficient (EPE) assignment. 

Similarly, in Profile 2 below, the set of envy-free assignments is as described
for some $0 \leq w \leq \frac{1}{6}$ and $0 \leq z \leq \frac{1}{12}$.
\[\begin{tabular}{llll}
\multicolumn{4}{l}{Profile 2} \\
1:& $o_1$, & $o_2$,& $o_3$\\
2:& $o_1$, & $\{o_2$& $o_3\}$\\  
3:& $\{o_1$ & $o_2\}$,& $o_3$\\ 
\end{tabular}
\quad 
\begin{tabular}{llll}
& $o_1$& $o_2$ & $o_3$ \\
1:& $\frac{1}{2}-w$ & $\frac{1}{4}+w-z$& $\frac{1}{4}+z$\\
2:& $\frac{1}{2}-w$ & $w+2z$& $\frac{1}{2}-2z$\\
3:& $2w$ & $\frac{3}{4}-2w-z$& $\frac{1}{4}+z$\\  
\end{tabular}
\]
Again, agent 2 cannot be assigned $o_2$ in any Pareto efficient assignment, as there is always a Pareto improvement with the agent assigned $o_3$, so $w = z = 0$ in
any ex-post efficient assignment. 

Observe that the properties of ex-post efficiency and envy-freeness determine a unique assignment in both Profile 1 and Profile 2. Furthermore, agents 1 and 2 have the same preferences in both profiles, but agent 3's allocation in Profile 1 stochastically dominates his allocation in Profile 2, implying a failure of weak strategyproofness.

For $n \geq 4$, extend each of the profiles as follows: the first 3 agents have exactly the same preference ordering over the first 3 objects; and they have strict preferences over the objects $o_4, o_5, \ldots, o_n$; finally, agent $i$ (for $i \geq 4$) is indifferent between the first $i$ objects, after which he has strict preferences over the others. That is, $i$'s preference ordering is
\[j: \{o_1 \ldots o_i\}, o_{i+1}, \ldots, o_n.\]
It is straightforward to check that agent $i$ receives object $o_i$ in every Pareto efficient assignment, and so the first 3 agents must be allocated the first 3 objects, leading to the same two profiles analyzed earlier.
\end{proof}

As the EPS mechanism is ordinally efficient (and so ex-post efficient as well) and envy-free, an immediate consequence is that the EPS mechanism is not weakly strategyproof on the uniform domain. The Random Priority (RP) mechanism, adapted to the setting of indifferences, is both strategyproof and ex-post efficient, and so fails envy-freeness. For the domain considered by Bogomolnaia and Moulin, neither of these results hold, as the PS mechanism is strategyproof and the RP mechanism is envy-free.

Next, we show that if we relax envy freeness to equal treatment of equals, but strengthen  weak strategyproofness and ex-post efficiency to strategyproofness and ordinal efficiency respectively, a similar impossibility result holds for the uniform preference domain.

\begin{theorem}
For $n \geq 4$, any mechanism that satisfies ordinal efficiency and equal treatment of equals is not strategyproof in the uniform preference domain.
\end{theorem}
\begin{proof}
We first show the result for $n=4$. Consider the following 8 profiles
\[\begin{tabular}{lllll}
\multicolumn{5}{l}{Profile 1} \\
1:& $o_1$, & $o_2$,& $o_3$,& $o_4$ \\
2:& $o_1$, & $o_2$,& $o_3$,& $o_4$ \\
3:& $o_1$, & $o_2$,& $o_3$,& $o_4$ \\  
4:& $o_1$, & $o_2$,& $o_3$,& $o_4$ \\ 
\end{tabular}
\quad 
\begin{tabular}{lllll}
\multicolumn{5}{l}{Profile 2} \\
1:& $o_1$, & $o_2$,& $o_3$,& $o_4$ \\
2:& $o_1$, & $o_2$,& $o_3$,& $o_4$ \\
3:& $o_1$, & $o_2$,& $o_3$,& $o_4$ \\  
4:& $\{o_1$ & $o_2\}$,& $o_3$,& $o_4$ \\ 
\end{tabular}
\]

\[
\begin{tabular}{lllll}
\multicolumn{5}{l}{Profile 3} \\
1:& $o_1$, & $o_2$,& $o_3$,& $o_4$ \\
2:& $o_1$, & $o_2$,& $o_3$,& $o_4$ \\
3:& $\{o_1$ & $o_2\}$,& $o_3$,& $o_4$ \\  
4:& $\{o_1$ & $o_2\}$,& $o_3$,& $o_4$ \\ 
\end{tabular}
\quad 
\begin{tabular}{lllll}
\multicolumn{5}{l}{Profile 4} \\
1:& $\{o_1$ & $o_2\}$,& $o_3$,& $o_4$ \\
2:& $o_1$, & $o_2$, & $o_3$,& $o_4$ \\
3:& $\{o_1$ & $o_2\}$,& $o_3$,& $o_4$ \\  
4:& $\{o_1$ & $o_2\}$,& $o_3$,& $o_4$ \\ 
\end{tabular}
\]

\[
\begin{tabular}{lllll}
\multicolumn{5}{l}{Profile 5} \\
1:& $o_1$, & $o_2$,& $o_3$,& $o_4$ \\
2:& $o_1$, & $\{o_2$ & $o_3\}$,& $o_4$ \\
3:& $o_1$, & $o_2$,& $o_3$,& $o_4$ \\  
4:& $o_1$, & $o_2$,& $o_3$,& $o_4$ \\ 
\end{tabular}
\quad
\begin{tabular}{lllll}
\multicolumn{5}{l}{Profile 6} \\
1:& $o_1$, & $o_2$,& $o_3$,& $o_4$ \\
2:& $o_1$, & $\{o_2$ & $o_3\}$,& $o_4$ \\
3:& $o_1$, & $o_2$,& $o_3$,& $o_4$ \\  
4:& $\{o_1$ & $o_2\}$,& $o_3$,& $o_4$ \\ 
\end{tabular}
\]

\[
\begin{tabular}{lllll}
\multicolumn{5}{l}{Profile 7} \\
1:& $o_1$, & $o_2$,& $o_3$,& $o_4$ \\
2:& $o_1$, & $\{o_2$& $o_3\}$,& $o_4$ \\
3:& $\{o_1$ & $o_2\}$,& $o_3$,& $o_4$ \\  
4:& $\{o_1$ & $o_2\}$,& $o_3$,& $o_4$ \\ 
\end{tabular}
\quad 
\begin{tabular}{lllll}
\multicolumn{5}{l}{Profile 8} \\
1:& $\{o_1$ & $o_2\}$,& $o_3$,& $o_4$ \\
2:& $o_1$, & $\{o_2$& $o_3\}$,& $o_4$ \\
3:& $\{o_1$ & $o_2\}$,& $o_3$,& $o_4$ \\  
4:& $\{o_1$ & $o_2\}$,& $o_3$,& $o_4$ \\ 
\end{tabular}
\]
We shall show that in profile $8$ there is no probabilistic assignment that simultaneously satisfies ordinal efficiency (OE), equal treatment of equals (ETE), and strategyproofness (SP) in relation to the first seven profiles. 

First, we compute the probability assignment for profile $1$. Notice that the only assignment that satisfies ETE is

\[\begin{tabular}{lllll}
\multicolumn{5}{l}{Profile 1} \\
1:& $o_1$, & $o_2$,& $o_3$,& $o_4$ \\
2:& $o_1$, & $o_2$,& $o_3$,& $o_4$ \\
3:& $o_1$, & $o_2$,& $o_3$,& $o_4$ \\  
4:& $o_1$, & $o_2$,& $o_3$,& $o_4$ \\ 
\end{tabular}
\quad 
\begin{tabular}{lllll}
& $o_1$& $o_2$ & $o_3$ & $o_4$ \\
1:& $\frac{1}{4}$ & $\frac{1}{4}$& $\frac{1}{4}$& $\frac{1}{4}$ \\
2:& $\frac{1}{4}$ & $\frac{1}{4}$& $\frac{1}{4}$& $\frac{1}{4}$ \\
3:& $\frac{1}{4}$ & $\frac{1}{4}$& $\frac{1}{4}$& $\frac{1}{4}$ \\  
4:& $\frac{1}{4}$ & $\frac{1}{4}$& $\frac{1}{4}$& $\frac{1}{4}$ \\ 
\end{tabular}
\]

Now we consider profile $2$. Let $p_{ij}$ be the probability that agent $i$ is assigned the object $o_j$. By ordinal efficiency, $p_{41} = 0$. For otherwise $p_{42} < 1$, which means that one of $p_{12}$, $p_{22}$, $p_{32}$ is strictly positive; this agent can exchange a small amount of $o_2$ for an equal amount of $o_1$ from agent 4, without altering any of the other allocations to obtain a new allocation matrix that stochastically dominates the current one, and this violates ordinal efficiency.

By strategyproofness, we must have that $p_{41} + p_{42} = \frac{1}{2}$, because if it were not the case, then there is a profitable deviation of agent $4$ either from profile $1$ to profile $2$ or vice versa. Thus, we get that $p_{42} = \frac{1}{2}$ since $p_{41} = 0$. Similarly, by strategyproofness, we have that $p_{41}+p_{42}+p_{43} = \frac{3}{4}$, which implies that $p_{43} = \frac{1}{4}$ and $p_{44} = \frac{1}{4}$. \\
Finally by ETE, we know the probability assignment of the first three agents must be identical, thus we get the following assignment:
\[
\begin{tabular}{lllll}
\multicolumn{5}{l}{Profile 2} \\
1:& $o_1$, & $o_2$,& $o_3$,& $o_4$ \\
2:& $o_1$, & $o_2$,& $o_3$,& $o_4$ \\
3:& $o_1$, & $o_2$,& $o_3$,& $o_4$ \\  
4:& $\{o_1$ & $o_2\}$,& $o_3$,& $o_4$ \\ 
\end{tabular}
\quad 
\begin{tabular}{lllll}
& $o_1$& $o_2$ & $o_3$ & $o_4$ \\
1:& $\frac{1}{3}$ & $\frac{1}{6}$& $\frac{1}{4}$& $\frac{1}{4}$ \\
2:& $\frac{1}{3}$ & $\frac{1}{6}$& $\frac{1}{4}$& $\frac{1}{4}$ \\
3:& $\frac{1}{3}$ & $\frac{1}{6}$& $\frac{1}{4}$& $\frac{1}{4}$ \\  
4:& $0$ & $\frac{1}{2}$& $\frac{1}{4}$& $\frac{1}{4}$ \\ 
\end{tabular}
\]

Now we consider profile $3$. In profile $3$, by SP in relation to profile $2$, we must have that $p_{31} + p_{32} = \frac{1}{2}$, $p_{33} = \frac{1}{4}$, and $p_{34} = \frac{1}{4}$.  By ETE, the same assignment for agent $4$ satisfy the same constraints as that of agent $3$. By OE, $p_{31} = p_{41} = 0$, because either $p_{31} > 0$ or  $p_{41} > 0$ would imply that $p_{32} + p_{42}<1$ (as $p_{31} + p_{32} + p_{41}+p_{42} = 1$)  or equivalently that $p_{12} + p_{22} > 0$. Then again we have a situation where agent $1$ or $2$ can exchange a small amount of $o_2$ for an equal amount of $o_1$ from agent $3$ or $4$, which leads to a new assignment matrix that stochastically dominates the current one, violating OE. Thus, OE and SP together determines the probabilistic assignment for agents $3$ and $4$. Now, we can fill in the assignments for agents $1$ and $2$ via ETE to get:
\[
\begin{tabular}{lllll}
\multicolumn{5}{l}{Profile 3} \\
1:& $o_1$, & $o_2$,& $o_3$,& $o_4$ \\
2:& $o_1$, & $o_2$,& $o_3$,& $o_4$ \\
3:& $\{o_1$ & $o_2\}$,& $o_3$,& $o_4$ \\  
4:& $\{o_1$ & $o_2\}$,& $o_3$,& $o_4$ \\ 
\end{tabular}
\quad 
\begin{tabular}{lllll}
& $o_1$& $o_2$ & $o_3$ & $o_4$ \\
1:& $\frac{1}{2}$ & $0$& $\frac{1}{4}$& $\frac{1}{4}$ \\
2:& $\frac{1}{2}$ & $0$ & $\frac{1}{4}$& $\frac{1}{4}$ \\
3:& $0$ & $\frac{1}{2}$& $\frac{1}{4}$& $\frac{1}{4}$ \\  
4:& $0$ & $\frac{1}{2}$& $\frac{1}{4}$& $\frac{1}{4}$ \\ 
\end{tabular}
\]

Now we consider profile $4$. By SP in relation with profile $3$ and ETE, we must have that $p_{11} + p_{12} = p_{31} + p_{32} = p_{41} + p_{42} = \frac{1}{2}$, $p_{13} = p_{33} = p_{43} = \frac{1}{4}$, and $p_{14} = p_{34} = p_{44} = \frac{1}{4}$. Since $p_{12}$+ $p_{32}$ + $p_{42} \leq 1$, in order to satisfy the unit demand for agents $1$, $3$, and $4$, we must have that at least one of $p_{11}$, $p_{31}$, $p_{41}$ is strictly positive. Thus by OE, we must have $p_{22} = 0$ and $p_{21} = \frac{1}{2}$. Although we cannot pin down a single assignment for this profile, any feasible assignment must be of the form:
\[\begin{tabular}{lllll}
\multicolumn{5}{l}{Profile 4} \\
1:& $\{o_1$ & $o_2\}$,& $o_3$,& $o_4$ \\
2:& $o_1$, & $o_2$, & $o_3$,& $o_4$ \\
3:& $\{o_1$ & $o_2\}$,& $o_3$,& $o_4$ \\  
4:& $\{o_1$ & $o_2\}$,& $o_3$,& $o_4$ \\ 
\end{tabular}
\quad
\begin{tabular}{lllll}
& $o_1$& $o_2$ & $o_3$ & $o_4$ \\
1:& $x$ & $\frac{1}{2}-x$& $\frac{1}{4}$& $\frac{1}{4}$ \\
2:& $\frac{1}{2}$ & $0$ & $\frac{1}{4}$& $\frac{1}{4}$ \\
3:& $y$ & $\frac{1}{2}-y$& $\frac{1}{4}$& $\frac{1}{4}$ \\  
4:& $\frac{1}{2}-x-y$ & $x+y$& $\frac{1}{4}$& $\frac{1}{4}$ \\ 
\end{tabular}
\]
for some $x, y \geq 0$ and $x+y \leq \frac{1}{2}$.

Now we consider profile $5$. Applying the same argument of ordinal efficiency for agent $4$ in profile $2$ to agent $2$ in profile $5$, we get that $p_{22} = 0$. By strategyproofness in relation to profile $1$, we must have that $p_{21} = \frac{1}{4}$ and that  $p_{21} + p_{22} + p_{23} = \frac{3}{4}$. Hence, we get that $p_{23} = \frac{1}{2}$ and $p_{24}= \frac{1}{4}$.
Finally by ETE, we know the probability assignment of the agents $1$, $2$ and $4$ must be identical, thus we get the following assignment:

\[\begin{tabular}{lllll}
\multicolumn{5}{l}{Profile 5} \\
1:& $o_1$, & $o_2$,& $o_3$,& $o_4$ \\
2:& $o_1$, & $\{o_2$ & $o_3\}$,& $o_4$ \\
3:& $o_1$, & $o_2$,& $o_3$,& $o_4$ \\  
4:& $o_1$, & $o_2$,& $o_3$,& $o_4$ \\ 
\end{tabular}
\quad 
\begin{tabular}{lllll}
& $o_1$& $o_2$ & $o_3$ & $o_4$ \\
1:& $\frac{1}{4}$ & $\frac{1}{3}$& $\frac{1}{6}$& $\frac{1}{4}$ \\
2:& $\frac{1}{4}$ & $0$ & $\frac{1}{2}$& $\frac{1}{4}$ \\
3:& $\frac{1}{4}$ & $\frac{1}{3}$& $\frac{1}{6}$& $\frac{1}{4}$ \\  
4:& $\frac{1}{4}$ & $\frac{1}{3}$& $\frac{1}{6}$& $\frac{1}{4}$ \\ 
\end{tabular}
\]

Now we consider profile $6$. By SP in relation with profile $2$, we must have that $p_{21} = \frac{1}{3}$, $p_{22} + p_{23} = \frac{5}{12}$, and $p_{24} = \frac{1}{4}$. By OE, we must have that $p_{22} = 0$, which implies that $p_{23} = \frac{5}{12}$. By SP in relation with profile $5$, we must have that $p_{41} + p_{42} = \frac{7}{12}$, $p_{43} = \frac{1}{6}$, $p_{44} = \frac{1}{4}$. Again, by OE, we must have that $p_{41} = 0$, which implies that $p_{42} = \frac{7}{12}$.  Subsequently, we can fill in the assignments for agents $1$ and $3$ via ETE to get:    
\[\begin{tabular}{lllll}
\multicolumn{5}{l}{Profile 6} \\
1:& $o_1$, & $o_2$,& $o_3$,& $o_4$ \\
2:& $o_1$, & $\{o_2$ & $o_3\}$,& $o_4$ \\
3:& $o_1$, & $o_2$,& $o_3$,& $o_4$ \\  
4:& $\{o_1$ & $o_2\}$,& $o_3$,& $o_4$ \\ 
\end{tabular}
\quad
\begin{tabular}{lllll}
& $o_1$& $o_2$ & $o_3$ & $o_4$ \\
1:& $\frac{1}{3}$ & $\frac{5}{24}$& $\frac{5}{24}$& $\frac{1}{4}$ \\
2:& $\frac{1}{3}$ & $0$ & $\frac{5}{12}$& $\frac{1}{4}$ \\
3:& $\frac{1}{3}$ & $\frac{5}{24}$& $\frac{5}{24}$& $\frac{1}{4}$ \\  
4:& $0$ & $\frac{7}{12}$& $\frac{1}{6}$& $\frac{1}{4}$ \\ 
\end{tabular}
\]

Now we consider profile $7$. By SP in relation with profile $3$, we must have that $p_{21} = \frac{1}{2}$, $p_{22} + p_{23} = \frac{1}{4}$, and $p_{24} = \frac{1}{4}$. By OE, we must have that $p_{22} = 0$, which implies that $p_{23} = \frac{1}{4}$. By SP in relation with profile $6$, we have that $p_{31} + p_{32} =\frac{13}{24}$, $p_{33} = \frac{5}{24}$ and $p_{34} = \frac{1}{4}$. By ETE, agent $4$ gets an equivalent assignment as agent $3$. Notice that in this case, we must have that $p_{31} > 0$ and $p_{41} > 0$ as $p_{32} + p_{42} \leq 1$ and $p_{31} + p_{32} + p_{41}+p_{42} = \frac{13}{12} > 1$, so either $p_{31}$ or $p_{41}$ is strictly positive. This implies that $p_{12} = p_{22} = 0$ in order to satisfy OE. Thus, we get an assignment of the following form:
\[
\begin{tabular}{lllll}
\multicolumn{5}{l}{Profile 7} \\
1:& $o_1$, & $o_2$,& $o_3$,& $o_4$ \\
2:& $o_1$, & $\{o_2$& $o_3\}$,& $o_4$ \\
3:& $\{o_1$ & $o_2\}$,& $o_3$,& $o_4$ \\  
4:& $\{o_1$ & $o_2\}$,& $o_3$,& $o_4$ \\ 
\end{tabular}
\quad
\begin{tabular}{lllll}
& $o_1$& $o_2$ & $o_3$ & $o_4$ \\
1:& $\frac{5}{12}$ & $0$& $\frac{1}{3}$& $\frac{1}{4}$ \\
2:& $\frac{1}{2}$ & $0$ & $\frac{1}{4}$& $\frac{1}{4}$ \\
3:& $z$ & $\frac{13}{24}-z$& $\frac{5}{24}$& $\frac{1}{4}$ \\  
4:& $\frac{1}{12} - z$ & $\frac{11}{24}+z$& $\frac{5}{24}$& $\frac{1}{4}$ \\ 
\end{tabular}
\]

Finally, we consider profile $8$. By SP in relation with profile $4$, we must have that $p_{21} = \frac{1}{2}$, $p_{22} + p_{23} = \frac{1}{4}$, $p_{24} = \frac{1}{4}$. By OE, we must have that $p_{22} = 0$, which implies that $p_{23} = \frac{1}{4}$. By SP in relation with profile $7$ and ETE, we must have that $p_{11} + p_{12} = p_{31} + x_{32} = x_{41} + x_{42} = \frac{5}{12}$, $p_{13} = p_{33} = p_{43} = \frac{1}{3}$, $p_{14} = p_{34} = p_{44} = \frac{1}{4}$. Now consider the partially filled assignment below
\[
\begin{tabular}{lllll}
\multicolumn{5}{l}{Profile 8} \\
1:& $\{o_1$ & $o_2\}$,& $o_3$,& $o_4$ \\
2:& $o_1$, & $\{o_2$& $o_3\}$,& $o_4$ \\
3:& $\{o_1$ & $o_2\}$,& $o_3$,& $o_4$ \\  
4:& $\{o_1$ & $o_2\}$,& $o_3$,& $o_4$ \\ 
\end{tabular}
\quad
\begin{tabular}{lllll}
& $o_1$& $o_2$ & $o_3$ & $o_4$ \\
1:& ? & $?$& $\frac{1}{3}$& $\frac{1}{4}$ \\
2:& $\frac{1}{2}$ & $0$ & $\frac{1}{4}$& $\frac{1}{4}$ \\
3:& ? & ?& $\frac{1}{3}$& $\frac{1}{4}$ \\  
4:& ? & ?& $\frac{1}{3}$& $\frac{1}{4}$ \\ 
\end{tabular}
\]
Notice that this assignment violates the fact that $p_{13} + p_{23} + p_{33} + p_{43} = 1$. Since we used the necessary conditions induced by SP, OE, ETE to pin down all possible assignments for each of the profiles $1$-$7$, and all of possible combinations of the first $7$ profiles lead to this contradiction, it is impossible to write down a random assignment in profile 8 that simultaneously satisfy ETE, OE, and SP in relation to the other $7$ profiles. 

The following graph indicates how the profiles are linked via strategyproofness. (Profile $i$ is denoted $P_i$.)

\begin{figure} [h]
	\begin{center}
		\begin{tikzpicture}[->,>=stealth',shorten >=1pt,auto,node distance=2.8cm,
		semithick]
		\tikzstyle{every state}=[fill=blue,draw=none,text=white]
		
		\node[state] (A)                    {$P_1$};
		\node[state]         (B) [above left of=A] {$P_2$};
		\node[state]         (C) [above left of=B] {$P_3$};
		\node[state]         (D) [above left of=C] {$P_4$};
		\node[state]         (E) [above right of=A] {$P_5$};
		\node[state]         (F) [above left of=E] {$P_6$};
		\node[state]         (G) [above left of=F] {$P_7$};
		\node[state]         (H) [above left of=G] {$P_8$};
		
		\path (A) edge              node {}  (B);	  
		\path (B) edge              node {}  (C);
		\path (B) edge              node {}  (F);
		\path (C) edge              node {}  (D);	  	  
		\path (C) edge              node {}  (G);
		\path (A) edge              node {}  (E);
		\path (E) edge              node {}  (F);
		\path (D) edge              node {}  (H);
		\path (F) edge              node {}  (G);
		\path (G) edge              node {}  (H);
		\end{tikzpicture}
	\end{center}
\end{figure}

For general $n \geq 5$, we extend each of the 8 profiles as follows: agents $1$ through $4$ have the same preference for objects $o_1$ through $o_4$; moreover, these agents have strict preference for the rest of the objects. For every $j = 5, \ldots, n$, agent $j$ is indifference amongst objects $o_1$ through $o_{j}$ and has strict preference for the rest of the objects. Similar to the argument made for general $n$ in theorem \ref{Imp: OE, EF wSP}, we see that by OE, every agent $j = 5, \ldots, n$ receives object $o_j$ with probability $1$.  Consequently, the first 4 agents must be allocated the first 4 objects, leading to the same 8 profiles analyzed earlier.

\end{proof}

\bibliographystyle{plain}
\bibliography{paper}

\end{document}